\documentclass{comjnl}

\usepackage{algorithm}
\usepackage{algorithmicx}
\usepackage{algpseudocode}
\usepackage{amsmath}
\allowdisplaybreaks
\usepackage{amssymb}
\usepackage{latexsym}
\usepackage{siunitx}
\usepackage{bm}
\usepackage{nicefrac}
\usepackage{booktabs}
\usepackage{array}
\usepackage{multirow}
\usepackage{threeparttable}
\usepackage{makecell}
\usepackage{stfloats}
\usepackage{graphicx}
\usepackage{color}

\usepackage{tabularx}
\usepackage{siunitx}

\newcommand\AAA{\boldsymbol{\mathit{A}}}
\newcommand\FF{\boldsymbol{\mathit{F}}}
\newcommand\RR{\boldsymbol{\mathit{R}}}
\newcommand\DD{\boldsymbol{\mathit{D}}}
\newcommand\LL{\boldsymbol{\mathit{L}}}
\newcommand\WW{\boldsymbol{\mathit{W}}}
\newcommand\II{\boldsymbol{\mathit{I}}}
\newcommand\avec{\boldsymbol{\mathit{a}}}
\newcommand\rankvec{\boldsymbol{\mathit{rank}}}
\newcommand\wvec{\boldsymbol{\mathit{w}}}
\newcommand\ee{\boldsymbol{\mathit{e}}}
\newcommand\SE{\equiv}
\newcommand\calF{\mathcal{F}}
\newcommand\calS{\mathcal{S}}

\newcommand\blfootnote[1]{%
	\begingroup
	\renewcommand\thefootnote{}\footnote{#1}%
	\addtocounter{footnote}{-1}%
	\endgroup
}

\begin{document}

\title[Role Similarity Metric Based on Spanning Rooted Forest]{Role Similarity Metric Based on Spanning Rooted Forest}
\author{Qi~Bao}\author{Zhongzhi~Zhang}

\author{Haibin~Kan\thanks{}}\email{\{20110240002, zhangzz, hbkan\}@fudan.edu.cn}

\affiliation{Shanghai Key Laboratory of Intelligent Information
	Processing, School of Computer Science, Fudan University, Shanghai 200433, China}

\shortauthors{Q. Bao, Z. Zhang, and H. Kan}



\begin{abstract}
As a fundamental issue in network analysis, structural node similarity has received much attention in academia and is adopted in a wide range of applications. Among these proposed structural node similarity measures, role similarity stands out because of satisfying several axiomatic properties including automorphism conformation. Existing role similarity metrics cannot handle top-k queries on large real-world networks due to the high time and space cost. In this paper, we propose a new role similarity metric, namely \textsf{ForestSim}. We prove that \textsf{ForestSim} is an admissible role similarity metric and devise the corresponding top-k similarity search algorithm, namely \textsf{ForestSimSearch}, which is able to process a top-k query in $O(k)$ time once the precomputation is finished. Moreover, we speed up the precomputation by using a fast approximate algorithm to compute the diagonal entries of the forest matrix, which reduces the time and space complexity of the precomputation to $O(\epsilon^{-2}m\log^5{n}\log{\frac{1}{\epsilon}})$ and $O(m\log^3{n})$, respectively. Finally, we conduct extensive experiments on 26 real-world networks. The results show that \textsf{ForestSim} works efficiently on million-scale networks and achieves comparable performance to the state-of-art methods.
\end{abstract}

\maketitle

\section{Introduction}
Structural node similarity is a fundamental issue in graph analysis~\cite{LuJa07, JeWi02, FuNa13, YuLi13, JiLe11, ZhTa15, XiGo19, WeDu19, MoZe16} and is adopted in a wide scope of applications. For example, node similarity has been used for role discovery~\cite{WaFa94, SeBo03, HeGa12}, protein function prediction~\cite{MiPr08, DaYa15}, anomaly detection~\cite{RoGa13, AkCh13}, recommender system~\cite{FoPi07, LiHu18, SyTi10, GuGu16}, and community detection~\cite{PaLi10, SaMo18}. Due to the vast applications of node similarity, many measures have been proposed to estimate the similarity between vertices~\cite{JiLe11,JeWi02,PaYa04,HeGa11}. Among these works, role similarity~\cite{JiLe14, JiLe11,LiQi15,ShLi19,YuIr20,ChLa20} stands out because of satisfying several axiomatic properties, including \textit{automorphism confirmation}, which ensures that the similarity score between two nodes that are automorphically equivalent is $1$ (the maximum similarity score). Thus, role similarity measures have the good merit of indicating automorphism, which keeps these measures from assigning counter-intuitive similarity scores.

\blfootnote{*Haibin Kan is  the corresponding author, who is also with Shanghai Engineering Research Institute of Blockchain, Shanghai 200433, China.}

\textsf{RoleSim}~\cite{JiLe11} and \textsf{StructSim}~\cite{ChLa20} are state-of-art role similarity metrics. However, these two measures have bottlenecks when finding the top-k similar nodes for a given node in a large network, which is the task that receives more interest from the end-users~\cite{LeLa12,ZhTa15,LiZh17,SuHa11,DiLi17,ZhTa17}. For an undirected unweighted graph with $n$ nodes and $m$ edges, \textsf{RoleSim} needs $O(n^2)$ space, which makes it unable to handle large graphs. Although \textsf{StructSim} only needs $O(nL\log{n})$ space, it still requires $O(nm)$ time to compute the indexes in the worst case. Moreover, \textsf{StructSim} needs $O(nL\log{n})$ time to process a top-k query, which is inefficient in large networks. Therefore, it is challenging and of great significance to propose a role similarity metric that only requires nearly linear time and space for precomputation and can finish a top-k similarity search in $O(k)$ time.

In this paper, on the basis of spanning rooted forests, we propose \textsf{ForestSim}, a new node similarity metric. \textsf{ForestSim} uses the average size of the trees rooted at the node $u$ in spanning rooted forests of the graph, denoted by $s(u)$, to capture its structural properties. Two nodes $u$ and $v$ are similar if $s(u)$ and $s(v)$ are close, while they are dissimilar if $s(u)$ and $s(v)$ are far apart. We further devise the corresponding top-k similarity search algorithm, namely \textsf{ForestSimSearch}. Once the precomputation is finished, \textsf{ForestSimSearch} can find the top-k similar nodes for a given node in $O(k)$ time. However, getting the whole forest matrix still needs $O(n^3)$ time and $O(n^2)$ space. In order to speed up the computation, we use the fast approximate algorithm to calculate the diagonal elements of the forest matrix, which can finish the precomputation within $O(\epsilon^{-2}m\log^5{n}\log{\frac{1}{\epsilon}})$ time and $O(m\log^3{n})$ space. Finally, we conduct extensive experiments on 26 real-world networks. Resluts show that \textsf{ForestSim} works as effective as other state-of-art role similarity metrics and is the only one that can process top-k queries on million-scale networks.

Our contributions are listed as follows.
\begin{itemize}
	\item[C1.] \textbf{New role similarity metric} : We propose \textsf{ForestSim}, a novel node similarity measure based on spanning rooted forests. We prove that \textsf{ForestSim} is an admissible role similarity metric.
	\item[C2.] \textbf{Efficient top-k similarity search algorithm} : We devise \textsf{ForestSimSearch} for the top-k similarity search. \textsf{ForestSimSearch} can handle a top-k query in $O(k)$ time once the precomputation is finished. Furthermore, we use the fast approximate algorithm to compute the diagonal entries of the forest matrix, which enables the precomputation to be finished within nearly linear time and nearly linear space.
	\item[C3.] \textbf{Extensive experimental studies} : We test the effectiveness of studied role similarity metrics on six labeled networks and estimate their efficiency on 20 real-world networks, including several large networks. For effectiveness, \textsf{ForestSim} achieves comparable performance to \textsf{RoleSim} and better results than \textsf{StructSim}. For efficiency, \textsf{ForestSim} is the unique role similarity metric that can process top-k queries within 4 hours on million-scale graphs.
\end{itemize}

\section{Related Work}~\label{S-RelatedWork}
In this section, we review two state-of-art role similarity metrics: \textsf{RoleSim} and \textsf{StructSim}.

\subsection{RoleSim}
Jin et al.~\cite{JiLe11} first defined the five axiomatic properties, including automorphism confirmation, that an admissible role similarity metric should satisfy. They also proposed \textsf{RoleSim}, the first admissible role similarity metric. \textsf{RoleSim} computes all-pairs of node similarity iteratively. The time complexity of computing \textsf{RoleSim} is very high since it requires a costly maximum matching during the computation. Moreover, computing \textsf{RoleSim} requires $O(n^2)$ space, which makes it unable to handle large real-world networks.

\subsection{StructSim}
The main drawback of \textsf{RoleSim} is the $O(n^2)$ space occupation. To solve this problem, Cheng et al.~\cite{ChLa20} proposed \textsf{StructSim}. \textsf{StructSim} that follows the hierarchical scheme and only requires $O(nL\log{n})$ space, where $L$ is the number of levels defined by users. Furthermore, once the precomputation is finished, it only takes $O(L\log{n})$ to compute the similarity score between a pair of given nodes, which makes it able to handle ad-hoc queries on billion-scale networks. However, \textsf{StructSim} is inefficient in handling top-k queries. To find the top-k similar nodes for a given node, \textsf{StructSim} needs to calculate all similarity scores between all other vertices and the given node, which takes $O(nL\log{n})$ time.

\section{Preliminary}~\label{S-Pre}
In this section, we briefly introduce some basic concepts about graph, role similarity metric, and spanning rooted forest.

\subsection{Graph and Related Matrices}
Let $G=(V,E)$ be an undirected unweighted simple graph, where $V=\{1,\dots, n\}$ is the vertex set, and $E=\{(u_1, v_1), \dots, (u_m, v_m)\}$, the edge set. Then, the graph $G$ contains $n = |V|$ vertices and $m = |E|$ edges. The connections of $G$ are encoded in its adjacency matrix $\AAA$, with $a_{ij}$ representing the adjacency relation between vertices $i$ and $j$. If $(i,j) \in E$, then $a_{ij} = a_{ji} = 1$; $a_{ij} = a_{ji} = 0$ otherwise. The degree of a vertex $i$ is $d_i = \sum_{j=1}^n a_{ij}$. The diagonal degree matrix of $G$ is defined as $\DD = {\rm diag}(d_1, d_2, \dots, d_n)$ and the Laplacian matrix is defined to be $\LL = \DD - \AAA$.

\subsection{Role Similarity Metric} 

To theoretically describe the role similarity measure, Jin \textit{et al.}~\cite{JiLe11} formulated a series of axiomatic properties that role similarity measures should obey.

Before introducing these axiomatic properties, we first review the concept of \textit{automorphism}. In a simple graph $G=(V,E)$, an automorphism is a bijection $f:V\rightarrow V$ such that for every edge $(u, v) \in E$, there is $(f(u), f(v)) \in E$. Two nodes $u,v \in V$ are automorphically equivalent, denoted by $u \SE v$, if and only if there exists an automorphism $f$ of $G$ satisfying $f(u) = v$. 

We next give the five axiomatic properties as follows.
\begin{definition}\label{def-axiom}
In a graph $G=(V,E)$, let $\textsf{Sim}(u, v)$ be the similarity score between any two nodes $u,v \in V$. Five axiomatic properties of role similarity are shown as follows:
\begin{enumerate}
	\item[P1.] Range: For any $u, v \in V$, there is $\textsf{Sim}(u, v) \in [0,1]$.
	\item[P2.] Symmetry: $\textsf{Sim}(u, v) = \textsf{Sim}(v, u)$ holds for any $u, v \in V$.
	\item[P3.] Automorphism confirmation: If $u \SE v$, there is $\textsf{Sim}(u, v) = 1$.
	\item[P4.] Transitive similarity: For three nodes $u, u', v \in V$, there is $\textsf{Sim}(u, v) = \textsf{Sim}(u', v)$ if $u \SE u'$.
	\item[P5.] Triangle inequality: For any $u, u', v \in V$, there is $\textsf{Dist}(u,v) \leq \textsf{Dist}(u,u') + \textsf{Dist}(u',v)$, where $\textsf{Dist}(u,v) = 1 - \textsf{Sim}(u,v)$.
\end{enumerate}
\end{definition}
If \textsf{Sim} obeys P1, P2, P3, and P4, it is an \textbf{admissible role similarity measure}. \textsf{Sim} is an \textbf{admissible role similarity metric} if it satisfies all five properties.

\begin{figure}
	\centering
	\includegraphics[width=0.9\linewidth]{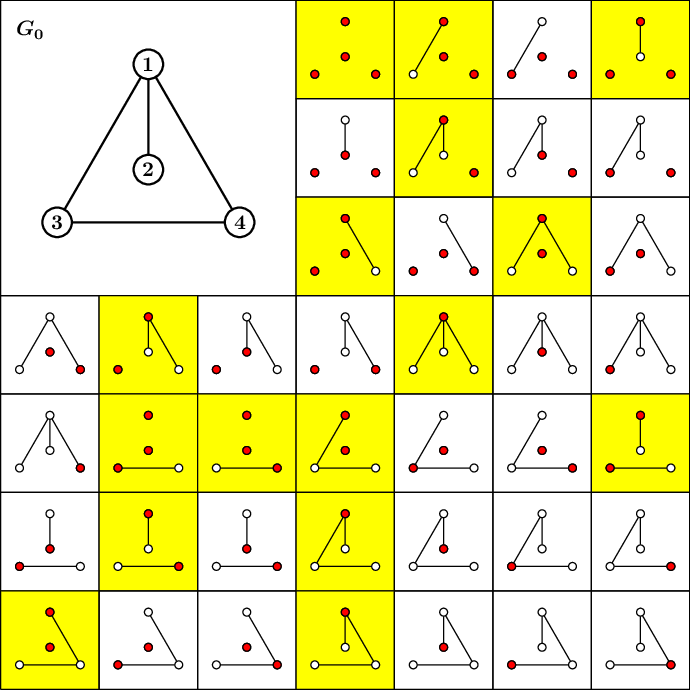}
	\caption{The toy graph $G_0$ and its all $40$ spanning rooted forests. $16$ forests in $\calF_{11}$ are marked in yellow.}\label{F1}
\end{figure}

\subsection{Spanning Rooted Forest}

Spanning rooted forests have pertained to multiple applications in graph mining~\cite{FoPi07, SeGa13, JiBaZh19, VaAn21}. For a graph $G=(V,E)$, a spanning forest $\mathcal{G}=(V',E')$ is an acyclic subgraph of $G$ with $V'=V$ and $E' \subseteq E$. Each connected component of $\mathcal{G}$ is a tree since the spanning forest has no graph cycles. A spanning rooted forest is a spanning forest with a root marked in each tree. Let $\calF = \calF(G)$ be the set containing all spanning rooted forests of $G$. Let $\calF_{ij} = \calF_{ij}(G)$ be the set including all spanning rooted forests of $G$, where nodes $i$ and $j$ are in the same tree rooted at $i$. Then, $\calF_{ii}$ is the set of all spanning rooted forests of $G$, where the node $i$ is a root. For example, for the graph $G_0 = (\{1,2,3,4\}, \{(1,2), (1,3), (1,4), (3,4)\})$ in Figure~\ref{F1}, it has $|\calF|=40$ spanning rooted forests, among which there are $|\calF_{11}|=16$ forests where the vertex $1$ is a root.

We then introduce the forest matrix, denoted by $\WW$~\cite{ChPa08, GoDr81}. The forest matrix is defined by $\WW = (\II + \LL)^{-1} = (w_{ij})_{n \times n}$, where $\II$ is the identity matrix. By definition, there is $|\calF| = |\II + \LL|$ and $w_{ij} = |\calF_{ij}|/|\calF|$~\cite{MeRu97, ZhXi11}. From Figure~\ref{F1}, the forest matrix of $G_0$ is 
\begin{align*}
	\WW = \frac{1}{40} \begin{pmatrix}
		16 & 8 & 8 & 8 \\
		8 & 24 & 4 & 4 \\
		8 & 4 & 19 & 9 \\
		8 & 4 & 9 & 19 \\
	\end{pmatrix}
\end{align*}
By definition, for each $w_{ij}$, there is $w_{ij} \in [0, 1]$.

We finally give some properties of diagonal elements in $\WW$.
\begin{lemma}\label{wuu-range}
	In a graph $G=(V,E)$, let $\WW$ be the forest matrix. Then, $w_{uu} \in \left[\frac{1}{d_u +1}, \frac{2}{d_u +2}\right]$ holds for any $u \in V$~\cite{ZhXi11}.
\end{lemma}

\begin{lemma}\label{wuu-auto}
In a simple graph $G=(V,E)$, let $\WW$ be the forest matrix. Then, $w_{uu}=w_{vv}$ holds for all $u \SE v$~\cite{BaZh21}.
\end{lemma}

\begin{figure}
	\centering
	\includegraphics[width=\linewidth]{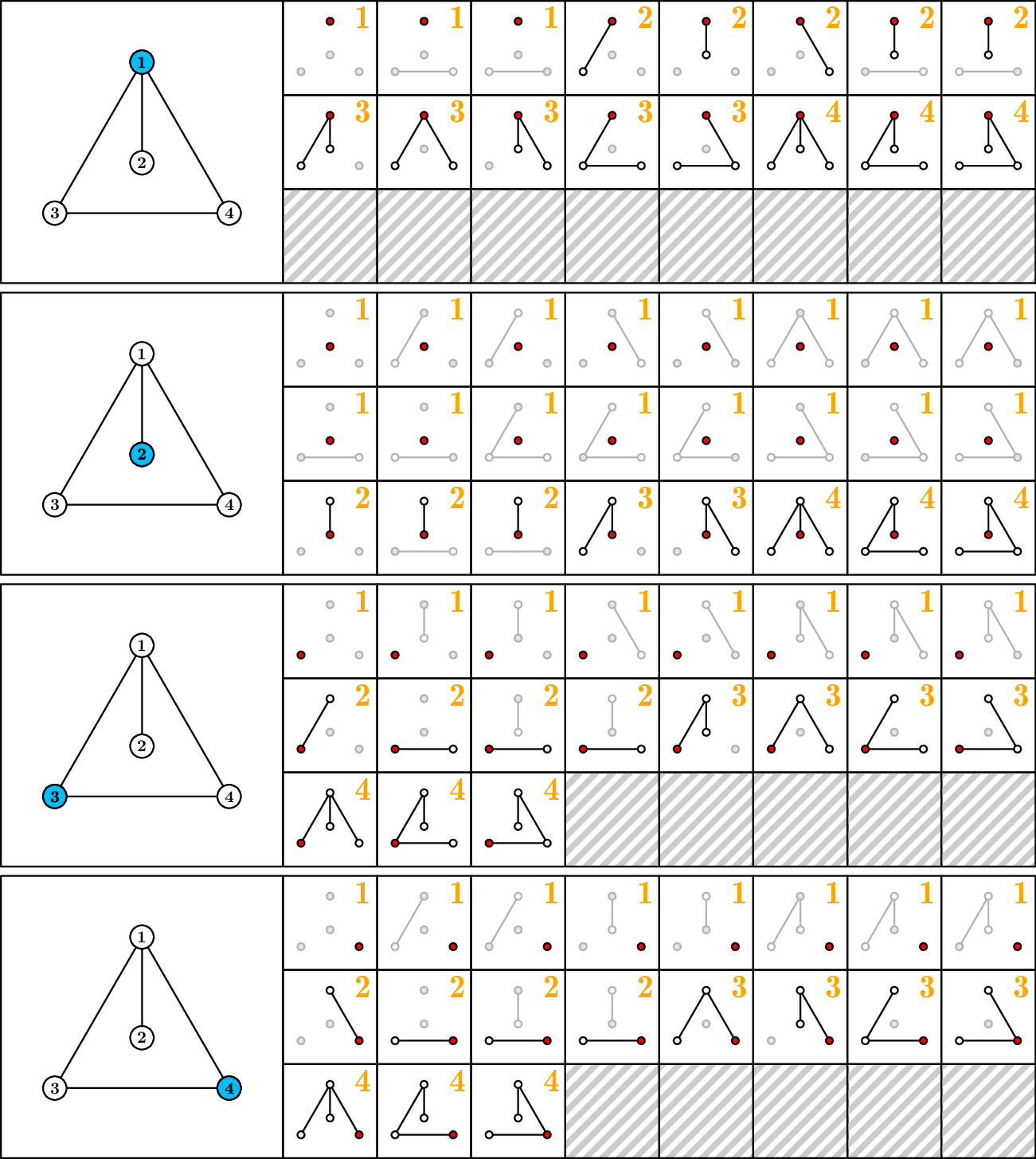}
	\caption{Each tree rooted at $u$ in the spanning rooted forest $F \in \calF_{uu}$ for $u= 1,2,3,4$ in the toy graph $G_0$. The size of each tree rooted at $u$, denoted by $|F|_u$, is labeled in the top right-hand corner of each picture.}\label{F2}
\end{figure}

\section{Role Similarity Metric Based on Spanning Rooted Forest}~\label{S-ForestSim}
In this section, we first define a new role similarity measure, namely \textsf{ForestSim}, based on spanning rooted forests. We then show that the \textsf{ForestSim} score can be expressed in terms of the diagonal elements in the forest matrix and prove that \textsf{ForestSim} is an admissible role similarity metric. After that, we propose \textsf{ForestSimSearch} to process the top-k queries. Finally, we speed up the computation by using a fast approximate algorithm to calculate the diagonal elements of the forest matrix.

\subsection{Definition}
The key point of analyzing structural roles is figuring out \textbf{how} a vertex connects with its context nodes~\cite{XuGe19}. To some extent, the sizes of those trees rooted at $u$ in the spanning rooted forests of a graph reflects the connection mode between the node $u$ and its context vertices. Here, we use the average size of all trees rooted at $u$ to capture the structural information of the vertex $u$.

\begin{definition}\label{defsu}
In a simple graph $G=(V,E)$, the average size of the trees rooted at $u$ in all spanning rooted forests of $G$, denoted by $s(u)$, is defined as
\begin{align}\label{eqsu}
	s(u) := \frac{1}{|\calF_{uu}|} \displaystyle\sum_{F \in \calF_{uu}} |F|_u,
\end{align}
where $|F|_u$ is the size of the tree rooted at $u$ in the forest $F$.
\end{definition}

In Figure~\ref{F2}, we present each tree rooted at $u$ and its size for $u=1,2,3$ and $4$, respectively. Thus, each $s(u)$ for $u \in \{1,2,3,4\}$ is 
\begin{align*}
	s(1) &= \frac{1\times 3 + 2 \times 5 + 3 \times 5 + 4\times 3}{3+5+5+3} = 2.5 \\
	s(2) &= \frac{1 \times 16 + 2 \times 3 + 3 \times 2 + 4 \times 3}{16+3+2+3} \approx 1.67 \\
	s(3) &= \frac{1 \times 8 + 2 \times 4 + 3 \times 4 + 4 \times 3}{8+4+4+3} \approx 2.11 \\
	s(4) &= \frac{1 \times 8 + 2 \times 4 + 3 \times 4 + 4 \times 3}{8+4+4+3} \approx 2.11 
\end{align*}
Note that two nodes $3$ and $4$ are symmetry in $G_0$, and there is $s(3) = s(4)$. Thus, for two similar nodes $u$ and $v$, $s(u)$ and $s(v)$ should be  close. Comes from this motivation, we define \textsf{ForestSim} as follows.

\begin{definition}
	Given a simple graph $G=(V,E)$ and two nodes $u,v \in V$, the \textsf{ForestSim} score between $u$ and $v$ is defined as 
\begin{align*}
	\textsf{ForestSim}(u, v) := \frac{\min(s(u), s(v))}{\max(s(u),s(v))},
\end{align*}
where $s(u)$ is defined by Definition~\ref{defsu}.
\end{definition}

Although the definition of \textsf{ForestSim} is succinct, it is impractical to calculate $s(u)$ by enumerating all spanning rooted forests in $\calF_{uu}$. We next propose an efficient way to calculate \textsf{ForestSim} by expressing $s(u)$ in terms of the elements in the forest matrix $\WW$.

\subsection{Computation}
We first evaluate the sum term $\sum_{F \in \calF_{uu}} |F|_u$ in~\eqref{eqsu}.

\begin{lemma}\label{lma-ssu}
In an undirected unweighted graph $G=(V,E)$, let $\calS_u$ be the set defined by 
\begin{align*}
	\calS_u = \{ (F, v^*) \; | \; F \in \calF_{uu}, v^* \in T(F, u) \},
\end{align*}
where $T(F, u)$ is the tree rooted at $u$ in $F$. Then, the following relation holds:
\begin{align*}
	\displaystyle\sum_{F \in \calF_{uu}} |F|_u = |\calS_u|.
\end{align*}
\end{lemma}

We next prove that the cardinality of the set $|\calS_u|$ is equal to the size of the set $\calF$.

\begin{figure}[h]
	\centering
	\includegraphics[width=0.8\linewidth]{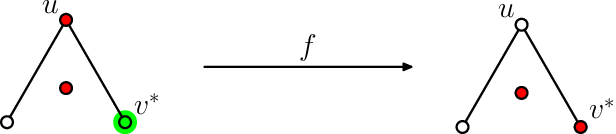}
	\caption{Construction of the mapping $f: \calS_u \rightarrow \calF$}\label{F3}
\end{figure}

\begin{lemma}\label{lma-su}
	In an undirected unweighted graph $G=(V,E)$, $|\calS_u|=|\calF|$ holds for any $u \in V$.
\end{lemma}

\begin{proof}
We prove this lemma by constructing a mapping $f:\calS_u \rightarrow \calF$ as follows. Let $(F_0, v_0^*)$ be an arbitrary element in $\calS_u$. The mapping $f$ can be considered as the following operation on the forest $F_0$. We let the node $v_0^*$ be the root of the tree containing $u$, which results in a spanning rooted forest $f(F_0, v_0^*) \in \calF$. Figure~\ref{F3} illustrates the construction of the mapping $f$.

We now prove that the mapping $f:\calS_u \rightarrow \calF$ is a one-to-one mapping. We first prove that $f$ is injective. For any two elements $(F_0, v_0^*) \neq (F_1, v_1^*)$, we distinguish two cases: $F_0 \neq F_1$ and $F_0 = F_1$. For the first case, $F_0 \neq F_1$, $f(F_0, v_0^*)$ and $f(F_1, v_1^*)$ have either different set of edges or different set of root nodes. For the second case, $F_0 = F_1$, we have $v_0^* \neq v_1^*$, with $v_0^*$ being a root in $f(F_0, v_0^*)$ but not a root in $f(F_1, v_1^*)$. Thus, for both cases, $f(F_0, v_0^*) \neq f(F_1, v_1^*)$, implying that $f$ is an injective mapping.

We proceed to prove that $f$ is surjective. In other words, for any $F' \in \calF$, there exists $F_0$ and $v_0^*$ satisfying $f(F_0, v_0^*) = F'$. For an arbitrary $F' \in \calF$, let $v_0^*$ be the root of the tree including the node $u$. We let $u$ be the root instead of $v_0^*$, which leads to a new forest $F_0$. Obviously, $f(F_0, v_0^*)=F'$, which indicates that $f$ is surjective. 

Therefore, there exists a one-to-one mapping $f:\calS_u \rightarrow \calF$, which indicates that $|\calS_u|=|\calF|$.
\end{proof}

Lemma~\ref{lma-su} leads to an efficient way to calculate \textsf{ForestSim}.

\begin{theorem}\label{thm-fs}
Let $G=(V,E)$ be an undirected unweighted graph with its forest matrix being $\WW$. Then, for any two vertices $u,v \in V$, the \textsf{ForestSim} score between $u$ and $v$ is 
\begin{align*}
	\textsf{\rm ForestSim}(u, v) = \frac{\min(w_{uu}, w_{vv})}{\max(w_{uu}, w_{vv})}.
\end{align*}
\end{theorem}

\begin{proof}
	By Lemma~\ref{lma-ssu} and~\ref{lma-su}, we have
\begin{align*}
	s(u) = \frac{1}{|\calF_{uu}|} \displaystyle\sum_{F \in \calF_{uu}} |F|_u = \frac{|\calS_u|}{|\calF_{uu}|} = \frac{|\calF|}{|\calF_{uu}|} = \frac{1}{w_{uu}}.
\end{align*}
Therefore,
\begin{align*}
	&\textsf{ForestSim}(u, v) = \frac{\min(s(u), s(v))}{\max(s(u),s(v))} \\
	=& \frac{\min(1/w_{uu}, 1/w_{vv})}{\max(1/w_{uu}, 1/w_{vv})} 
	= \frac{\min(w_{uu}, w_{vv})}{\max(w_{uu}, w_{vv})},
\end{align*}
which completes the proof.
\end{proof}

Theorem~\ref{thm-fs} shows that \textsf{ForestSim} score between two nodes can be represented in terms of diagonal elements of the forest matrix $\WW$. Thus, the \textsf{ForestSim} score for all pairs of nodes in a graph with $n$ vertices can be exactly calculated in $O(n^3)$.

\subsection{Admissibility}
We next prove the admissibility of \textsf{ForestSim}.

\begin{theorem}
\textsf{ForestSim} is an admissible role similarity metric.
\end{theorem}

\begin{proof}
Trivially, \textsf{ForestSim} obeys Range (P1) and Symmetry (P2). It has been proved that Transitive similarity (P4) is implied by Triangle inequality (P5)~\cite{JiLe11}. Thus, we only need to prove that \textsf{ForestSim} satisfies Automorphism conformation (P3) and Triangle inequality (P5).

Since $u \SE v$, there is $w_{uu} = w_{vv}$ based on Lemma~\ref{wuu-auto}. Hence, 
\begin{align*}
	\textsf{ForestSim}(u, v) = \frac{\min(w_{uu},w_{vv})}{\max(w_{uu},w_{vv})} = \frac{w_{uu}}{w_{uu}} = 1.
\end{align*}
Therefore, \textsf{ForestSim} obeys Automorphism conformation.

According to Definition~\ref{def-axiom}, there is
\begin{align}
	& \textsf{Dist}(u, v) = 1 - \textsf{ForestSim}(u, v) \notag\\
	=& 1 - \frac{\min(w_{uu}, w_{vv})}{\max(w_{uu}, w_{vv})} = \frac{|w_{uu} - w_{vv}|}{\max(w_{uu}, w_{vv})}. \label{eq:dist}
\end{align}
Since for any three nodes $u, u'$ and $v$ in a graph, there is $w_{uu},w_{u'u'},w_{vv} > 0$ based on Lemma~\ref{wuu-range}. Thus, according to~\cite{TiPa09}, we have
\begin{align*}
	\frac{|w_{uu}-w_{vv}|}{\max(w_{uu},w_{vv})} \leq \frac{|w_{uu}-w_{u'u'}|}{\max(w_{uu}, w_{u'u'})} + \frac{|w_{u'u'}-w_{vv}|}{\max(w_{u'u'}, w_{vv})}.
\end{align*} 
That is,
\begin{align*}
	\textsf{Dist}(u,v) \leq \textsf{Dist}(u,u') + \textsf{Dist}(u',v),
\end{align*}
which indicates that \textsf{ForestSim} satisfies Triangle inequality.

Consequently, \textsf{ForestSim} satisfies all five properties listed in Definition~\ref{def-axiom}. In other words, \textsf{ForestSim} is an admissible role similarity metric.
\end{proof}

\subsection{Top-k Similarity Search Algorithm}
In this subsection, we propose a fast top-k similarity search algorithm, namely \textsf{ForestSimSearch}, to handle the top-k query. Like \textsf{StructSim}, \textsf{ForestSimSearch} also needs precomputation. 

In the precomputation phase, we first compute the forest matrix and store all its diagonal elements in the list $\wvec$. Then, we sort the vertices in the ascending order of their corresponding diagonal elements in the forest matrix and store the result and the rank in $\avec$ and $\rankvec$ respectively.

Therefore, the following lemma holds.
\begin{lemma}\label{lem:precompute}
In a simple graph $G=(V,E)$, let $\WW=(w_{ij})_{n \times n}$ be its forest matrix. Three lists $\wvec$, $\avec$, and $\rankvec$ are defined as follows.
\begin{align}
	\wvec &:= [w_{11}, w_{22}, \dots, w_{nn}] \label{def:wvec} \\
	\avec &:= [a_1,a_2,\dots,a_n] \;\;\;\; w_{a_i a_i} \leq w_{a_{i+1} a_{i+1}} \label{def:avec} \\
	\rankvec &:= [rank_1, rank_2, \dots, rank_n] \;\;\;\; a_{rank_i} = i \label{def:rankvec}
\end{align}
Then, for any $u_1,u_2,u_3 \in V$ satisfying $rank_{u_1} < rank_{u_2} < rank_{u_3}$, the following statements hold.
\begin{align*}
	\textsf{\rm ForestSim}(u_1, u_2) & \geq \textsf{\rm ForestSim}(u_1, u_3) \\
	\textsf{\rm ForestSim}(u_2, u_3) & \geq \textsf{\rm ForestSim}(u_1, u_3)
\end{align*}
\end{lemma}

\begin{proof}
According to~\eqref{def:avec} and~\eqref{def:rankvec}, we have 
\begin{align*}
	w_{u_1 u_1} \leq w_{u_2 u_2} \leq w_{u_3 u_3}.
\end{align*}	
Thus, based on~\eqref{eq:dist}, there is
\begin{align*}
	& \textsf{Dist}(u_1, u_2) = \frac{w_{u_2 u_2} - w_{u_1 u_1}}{w_{u_2 u_2}} = 1 - \frac{w_{u_1 u_1}}{w_{u_2 u_2}} \\
	\leq & 1 - \frac{w_{u_1 u_1}}{w_{u_3 u_3}}  
	 = \frac{w_{u_3 u_3} - w_{u_1 u_1}}{w_{u_3 u_3}} = \textsf{Dist}(u_1, u_3). \\
\end{align*}
Due to the fact that $\textsf{Dist}(u,v) = 1 - \textsf{ForestSim}(u,v)$, we have 
\begin{align*}
	\textsf{ForestSim}(u_1, u_2) \geq \textsf{ForestSim}(u_1, u_3).
\end{align*}
Similarly, we can prove that 
\begin{align*}
	\textsf{ForestSim}(u_2, u_3) \geq \textsf{ForestSim}(u_1, u_3),
\end{align*}
which completes the proof.
\end{proof}

According to Lemma~\ref{lem:precompute}, for any node $u$ in the graph, we have
\begin{align*}
	&\textsf{ForestSim}(a_1, u) \leq \textsf{ForestSim}(a_2, u)  \leq \cdots \\
	 \cdots \leq &\textsf{ForestSim}(a_{rank_u-2}, u) \leq \textsf{ForestSim}(a_{rank_u-1}, u) \\
\end{align*}
and
\begin{align*}
	&\textsf{ForestSim}(a_n, u) \leq \textsf{ForestSim}(a_{n-1}, u)  \leq \cdots \\
	 \cdots \leq &\textsf{ForestSim}(a_{rank_u+2}, u) \leq \textsf{ForestSim}(a_{rank_u+1}, u). \\
\end{align*}
Therefore, inspired by merge sort~\cite{CoLe09}, we devise the top-k similarity search algorithm, namely \textsf{ForestSimSearch}, detailed in Algorithm~\ref{alg:ForestSimSearch}. Obviously, \textsf{ForestSimSearch} runs in $O(k)$ time. 

\begin{algorithm}
\caption{$\textsc{ForestSimSearch}(u,k)$}\label{alg:ForestSimSearch}
\begin{algorithmic}[1]
    \Require $G$, $\wvec$, $\avec$, $\rankvec$
    \Procedure{ForestSimSearch}{$u,k$}
	\State $l = rank_u-1$, $r = rank_u+1$, $S_{u,k} = \emptyset$ 
	\For {$i=1$ to $k$} 
		\If {$\textsf{\rm ForestSim}(a_l,u)<\textsf{\rm ForestSim}(a_r,u)$}
			\State $S_{u,k} = S_{u,k} \cup \{a_l\}$ 
			\State $l = l-1$ 
		\ElsIf 
			\State $S_{u,k} = S_{u,k} \cup \{a_r\}$ 
			\State $r = r+1$ 
		\EndIf
	\EndFor
	\State \textbf{return} $S_{u,k}$
    \EndProcedure 
\end{algorithmic}
\end{algorithm}

\subsection{Fast Approximate Algorithm}
The key step to compute \textsf{ForestSim} is calculating the diagonal elements of $\WW$. It takes $O(n^3)$ time and $O(n^2)$ space to get all diagonal elements of $\WW$ by inverting matrix $\II+\LL$, and we call this method \textsf{ForestSim-EX}. High time and space complexity of \textsf{ForestSim-EX} hinders its application to large real-world networks. To reduce the time and space complexity, we use a fast approximate algorithm to calculate the diagonal elements of the forest matrix, and accordingly call the algorithm \textsf{ForestSim-AP}.

Jin et al.~\cite{JiBaZh19} proposed a fast approximate algorithm, based on Johnson–Lindenstrauss lemma~\cite{JoLi84, Ac03} and an efficient linear system solver~\cite{KySa16}, to compute the diagonal entries of the forest matrix. Given a graph $G=(V,E)$ with $n$ vertices and $m$ edges, and an error parameter $\epsilon \in (0,0.5)$, the fast approximate algorithm~\cite{JiBaZh19} returns all $\tilde{w}_{uu}$ for all $u \in V$. For each node $u$, there is 
\begin{align*}
	(1-\epsilon) w_{uu} \leq \tilde{w}_{uu} \leq (1+\epsilon) w_{uu}
\end{align*}
with the probability $1-1/n$. Calculating the approximate values for all diagonal entries of $\WW$ takes $O(\epsilon^{-2}m\log^5{n}\log{\frac{1}{\epsilon}})$ time and $O(m\log^3{n})$ space~\cite{KySa16}. 

Here, we simply set $w_{uu}=\tilde{w}_{uu}$, which already reaches satisfactory results. Once we obtain the diagonal entries of the forest matrix, it only takes $O(n\log{n})$ time and $O(n)$ space to compute $\avec$ and $\rankvec$. Therefore, \textsf{ForestSim-AP} can finish the precomputation within nearly linear time and nearly linear space.

\begin{figure*}
	\centering
	\includegraphics[width=0.9\linewidth]{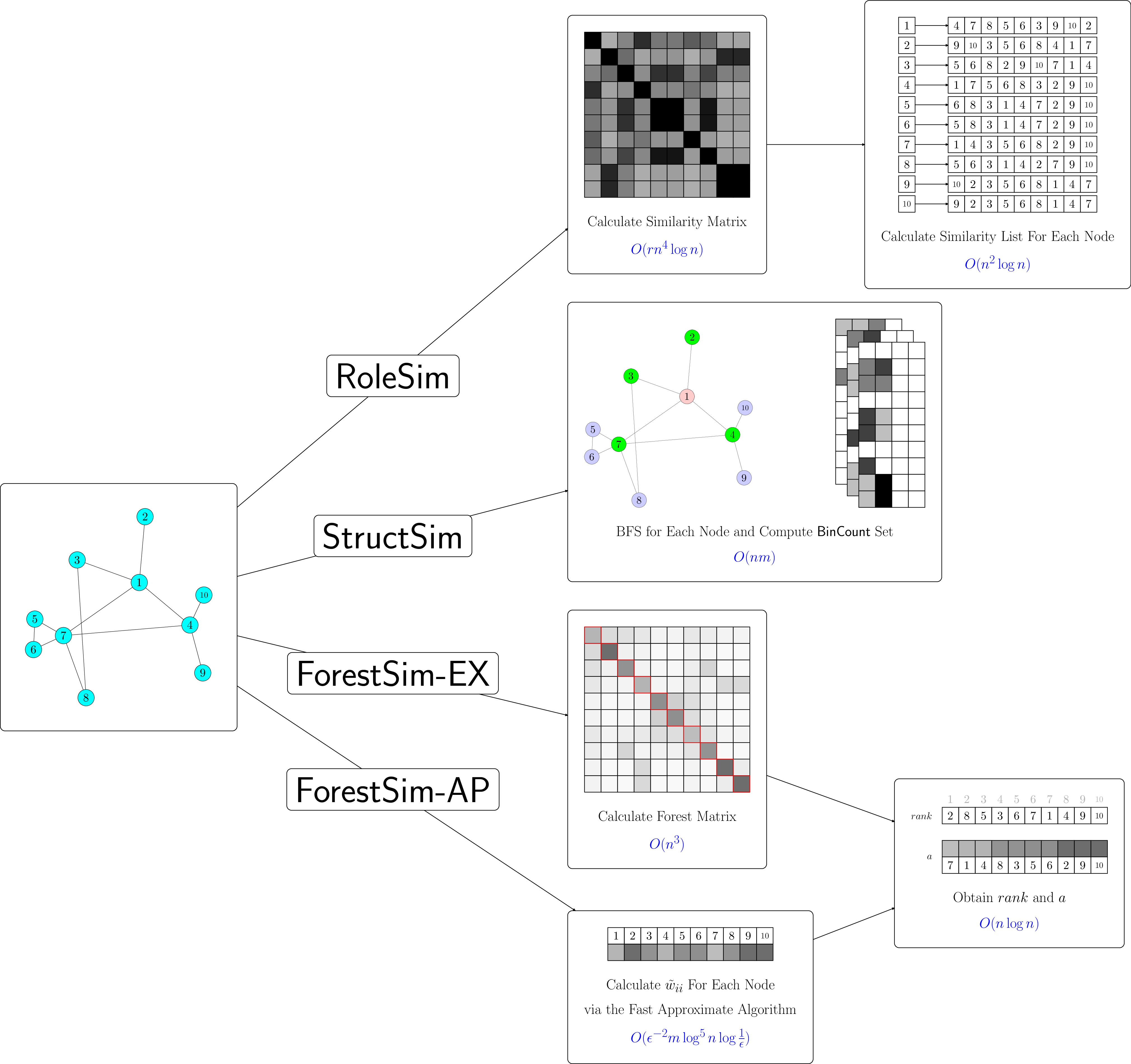}
	\caption{Precomputation of the studied role similarity metrics.}\label{F4}
\end{figure*}

\begin{table*}
	\centering
	\caption{Time and space complexity of the studied role similarity metrics}\label{Tab1}
	\begin{tabular}{lccrcrccr}
\toprule
\multirow{2}*{Metric} &&& \multicolumn{3}{c}{Time Complexity (Worst Case)} &&& \multirow{2}*{Space Complexity} \\
\cmidrule{4-6}
&&& Precomputation && Top-K Similarity Search &&& \\
\midrule
\textsf{RoleSim} &&& $O(rn^4\log{n})$ && $O(k)$ &&& $O(n^2)$ \\
\textsf{StructSim} &&& $O(nm)$ && $O(nL\log{n})$ &&& $O(m + nL\log{n})$ \\
\textsf{ForestSim-EX} &&& $O(n^3)$ && $O(k)$ &&& $O(n^2)$ \\
\textsf{ForestSim-AP} &&& $O(\epsilon^{-2}m\log^5{n}\log{\frac{1}{\epsilon}})$ && $O(k)$ &&& $O(m\log^3{n})$ \\
\bottomrule
	\end{tabular}
\end{table*}

\section{Complexity Analysis}~\label{S-Complex}
In this section, we systematically analyze the time and space complexity of \textsf{RoleSim}~\cite{JiLe11}, \textsf{StructSim}~\cite{ChLa20}, and \textsf{ForestSim} in finding top-k similar nodes for a given node. For each role similarity measure, its time complexity includes two parts: precomputation and top-k similarity search. Precomputation of the studied role similarity metrics is shown in Figure~\ref{F4}, and we list the time and space complexity of each measure in Table~\ref{Tab1}.

We first analyze \textsf{RoleSim}. In the precomputation phase, one first calculates the \textsf{RoleSim} score matrix $\RR$, which needs $O(rn^4\log{n})$ time where $r$ is the number of iterations. Then, for each node $u$, we can compute its similarity list $[v_1,v_2,\dots,v_{n-1}]$ satisfying $\textsf{RoleSim}(u,v_i) \geq \textsf{RoleSim}(u,v_{i+1})$ for $i = 1,2,\dots,n-2$, which requires $O(n^2\log{n})$ time.   For an arbitrary node $u$ in $G$, one only needs to output the first $k$ element in its similarity list to obtain the top-k similar nodes for node $u$, which runs in $O(k)$ time. Therefore, it takes $O(rn^4\log{n})$ time to finish the precomputation, and $O(k)$ time to process a top-k query. For space complexity, it requires $O(n^2)$ space to compute $\RR$ and store similarity lists of all vertices.

For \textsf{StructSim}, we only consider \textsf{StructSim-BC} in~\cite{ChLa20} with $L$ being the number of levels defined by the user. In the precomputation phase, it computes the \textsf{BinCount} set by doing breadth-first search for each node in the graph, which takes $O(nm)$ time in the worst case. To find top-k similar nodes for a given node $u$, one needs to calculate the similarity scores between the node $u$ and all other nodes, which takes $O(nL\log{n})$ time. After calculating these similarity scores, one also needs to sort to obtain the top-k similar nodes, which takes $O(n\log{n})$ time. Hence, it requires $O(nm)$ time to finish the precomputation, and $O(nL\log{n})$ time to handle a top-k query. The space complexity of \textsf{StructSim-BC} is $O(m+nL\log{n})$ since it requires $O(m)$ space to store the adjacent list of the graph, and $O(nL\log{n})$ space to store the \textsf{BinCount} set.

We proceed to analyze \textsf{ForestSim-EX}. In the precomputation phase, we first obtain all diagonal entries of the forest matrix by calculating $\WW$, which takes $O(n^3)$ time. Then, it requires $O(n\log{n})$ time to compute $\avec$ and $\rankvec$. We use \textsf{ForestSimSearch} to find the top-k similar vertices, which takes $O(k)$ time. Thus, it takes $O(n^3)$ time to finish the precomputation, and $O(k)$ time to process a top-k query. The space complexity of \textsf{ForestSim-EX} is $O(n^2)$ since it requires $O(n^2)$ space to store $\WW$ and $O(n)$ space to store $\wvec$, $\avec$, and $\rankvec$.

We finally analyze \textsf{ForestSim-AP}. In the precomputation phase, it uses the fast approximate algorithm~\cite{JiBaZh19} to compute the approximate values of all diagonal elements of the forest matrix, which requires $O(\epsilon^{-2}m\log^5{n}\log{\frac{1}{\epsilon}})$ time~\cite{JoLi84, KySa16}. Computing $\avec$ and $\rankvec$ takes $O(n\log{n})$ time. The top-k similarity search of \textsf{ForestSim-AP} is the same as that of \textsf{ForestSim-EX}, which takes $O(k)$ time. Therefore, it takes $O(\epsilon^{-2}m\log^5{n}\log{\frac{1}{\epsilon}})$ time to finish the precomputation, and $O(k)$ time to handle a top-k query. The space complexity of \textsf{ForestSim-AP} is $O(m\log^3{n})$ for the solver needs $O(m\log^3{n})$ space~\cite{KySa16}, and storing $\wvec$, $\avec$, and $\rankvec$ requires $O(n)$ space.

\textsf{RoleSim} and \textsf{ForestSim-EX} cannot work on large graphs since they both require $O(n^2)$ space, while \textsf{StructSim} and \textsf{ForestSim-AP} only need nearly linear space, which makes them able to run on large graphs. However, \textsf{StructSim} needs $O(nm)$ time to finish the precomputation, while \textsf{ForestSim-AP} only needs nearly linear time. Moreover, it only takes $O(k)$ time for \textsf{ForestSim-AP} to process a top-k query, while \textsf{StructSim} needs $O(nL\log{n})$, which is unacceptable on large networks. In brief, \textsf{ForestSim-AP} is the unique measure that can efficiently handle top-k queries on large real-world networks. 

\begin{table}
	\centering
	\caption{Detailed statistics of real-world networks used for effectiveness evaluation.}\label{exp-data-info1}
	\begin{tabular}{lcrrcr}
		\toprule
Network && $n$ & $m$ && Labels \\
		\midrule
Karate && 34 & 77 && 3 \\
Brazil && 131 & 1,003 && 4 \\
Europe && 399 & 5,993 && 4 \\
Expressway-status && 348 & 675 && 2 \\
Expressway-starbucks && 348 & 675 && 3 \\
Actor && 7,779 & 26,752 && 4 \\
		\bottomrule
	\end{tabular}
\end{table}

\section{Experiment}~\label{S-Exp}
In this section, we compare \textsf{ForestSim} with \textsf{RoleSim} and \textsf{StructSim} in terms of both effectiveness and efficiency. We implement \textsf{ForestSim}, \textsf{RoleSim} and \textsf{StructSim} in Julia (version 1.5.3). Our code is available online~\footnote{https://github.com/Accelerator950113/ForestSim}. For \textsf{ForestSim-AP}, we set the error parameter $\epsilon=0.1$. For \textsf{StructSim}, we set the number of levels $L=3$. All experiments are conducted on a machine configured with a 4-core 4.2GHz Intel i7-7700K CPU and 32GB RAM.

\begin{table*}
	\centering
	\caption{Statistics of large real networks used in experiments and running time (second, s) of four studied measures.}\label{exp-time}
	\begin{tabular}{lrrcrrrr}
		\toprule
\multirow{2}*{Network} & \multirow{2}*{$n$} & \multirow{2}*{$m$} && \multicolumn{4}{c}{Running Time (s)} \\
		\cmidrule{5-8}
&&&& \textsf{RoleSim} & \textsf{StructSim} & \textsf{ForestSim-EX} & \textsf{ForestSim-AP} \\
		\midrule
Diseasome & 516 & 1,188 &  & 28.07 & 0.32 & 0.04 & 2.91\\
EmailUniv & 1,133 & 5,451 &  & 749.40 & 1.48 & 0.17 & 1.25\\
Hamster & 2,426 & 16,630 &  & 7816.60 & 7.91 & 0.91 & 3.84\\
GridWorm & 3,507 & 6,531 &  & 1358.08 & 16.51 & 2.06 & 2.61\\
GrQc & 4,158 & 13,422 &  & 4568.81 & 24.71 & 2.99 & 5.29\\
Erdos992 & 5,094 & 7,515 &  & 1944.24 & 36.16 & 4.63 & 3.17\\
Reality & 6,809 & 7,680 &  & 2380.87 & 63.88 & 10.17 & 2.60\\
Dmela & 7,393 & 25,569 &  & 22392.88 & 82.98 & 16.95 & 10.66\\
HepPh & 11,204 & 117,619 &  & --- & 194.00 & 43.11 & 67.07\\
PagesCompany & 14,113 & 52,126 &  & --- & 298.49 & 82.41 & 50.65\\
AstroPh & 17,903 & 196,972 &  & --- & 523.13 & 162.16 & 135.71\\
Gplus & 23,628 & 39,194 &  & --- & 883.34 & 354.28 & 29.13\\
Brightkite & 56,739 & 212,945 &  & --- & 4468.01 & --- & 259.99\\
BlogCatalog & 88,784 & 2,093,195 &  & --- & 13589.79 & --- & 1318.71\\
Douban & 154,908 & 327,162 &  & --- & 45183.37 & --- & 504.17\\
MathSciNet & 332,689 & 820,644 &  & --- & --- & --- & 1965.01\\
Flickr & 513,969 & 3,190,452 &  & --- & --- & --- & 5008.16\\
IMDB & 896,305 & 3,782,447 &  & --- & --- & --- & 10236.09\\
YoutubeSnap & 1,134,890 & 2,987,624 &  & --- & --- & --- & 7897.06\\
Flixster & 2,523,386 & 7,918,801 &  & --- & --- & --- & 13051.02\\
		\bottomrule
	\end{tabular}
\end{table*}

\subsection{Efficiency}
In this subsection, we evaluate the efficiency of studied metrics on 20 real networks, all of which are publicly available in Koblenz Network Collection~\cite{Ku13} and Network Repository~\cite{RyNe15}. These real-world networks are preprocessed as simple graphs. In each network, we find the top-10 similar nodes for all vertices. Statistics of each network and the experimental results are listed in Table~\ref{exp-time}. 

Due to the high time complexity, \textsf{RoleSim} cannot finish the computation in 24 hours on the graph that contains more than 10,000 vertices. For \textsf{StructSim}, it fails to finish the computation in 24 hours on the networks including more than 300,000 vertices since it needs $O(n^2L\log{n})$ time to find the top-10 similar vertices for all nodes in the graph. For \textsf{ForestSim-EX}, it cannot handle the networks with over 50,000 nodes because it requires $O(n^2)$ space. \textsf{ForestSim-AP} works well on all 20 real-world networks. Moreover, it only takes about four hours for ForestSim-AP to finish the computation on Flixster, a graph with over 2,000,000 nodes.

Note that for large networks, \textsf{RoleSim}, \textsf{StructSim}, and \textsf{ForestSim-EX} cannot finish the computation due to their high time and memory cost, while \textsf{ForestSim-AP} works well on all these real networks, which shows the significant efficiency advantages of \textsf{ForestSim-AP}.

\begin{figure*}
	\centering
	\includegraphics[width=\linewidth]{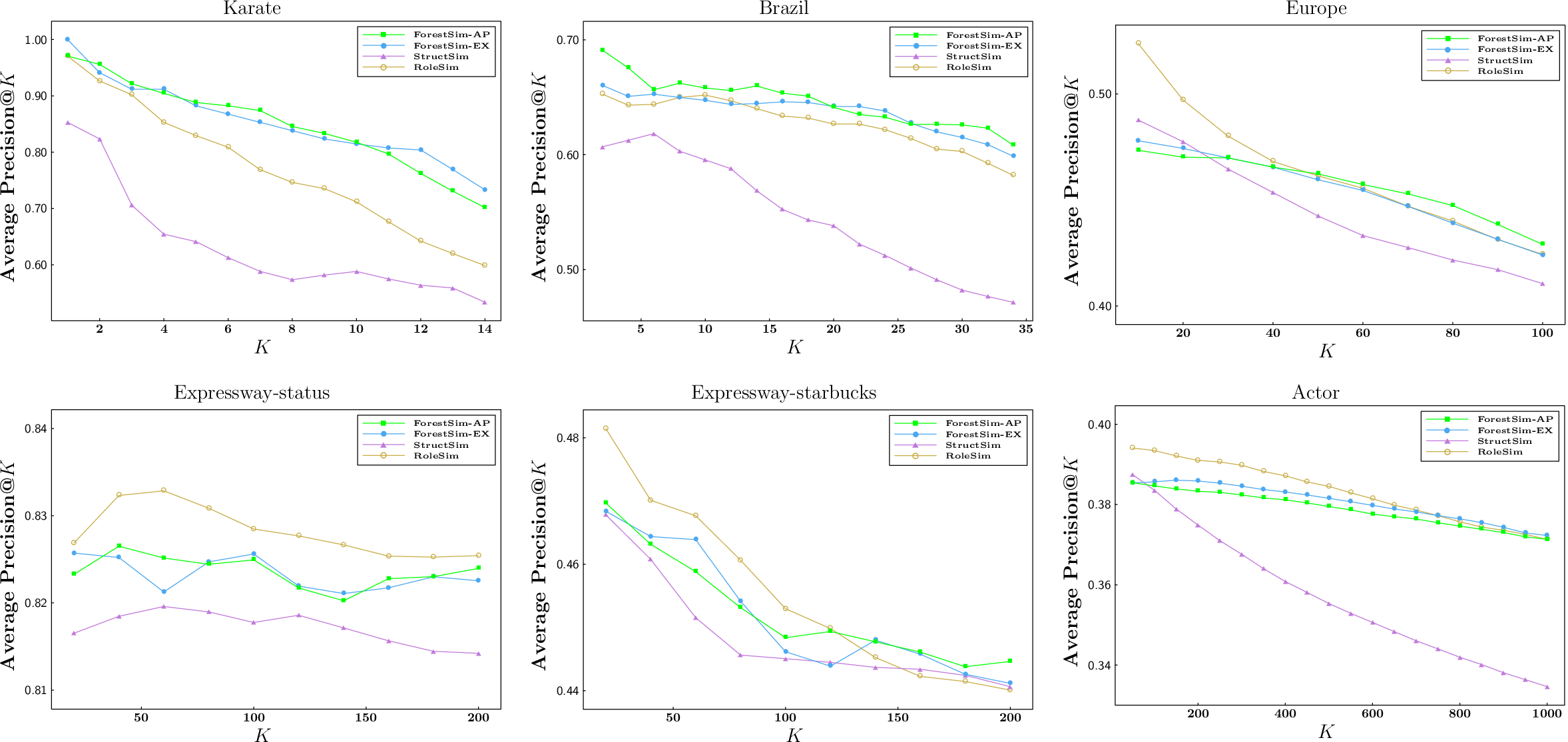}
	\caption{Average Precision@$K$ of the studied role similarity metrics on six real-world networks.}\label{AvgPreK}
\end{figure*}

\subsection{Effectiveness}
We use six labeled networks~\cite{TaSu09, RiSa17, XuGe19} to evaluate the effectiveness of the studied measures. These datasets are extensively used in machine learning studies~\cite{RiSa17, XuGe19, DoZi18}. Each vertex in the datasets has a label related to its topology, and we thus use the labels as the ground truth of roles. Related information about these real-world networks is listed in Table~\ref{exp-data-info1}.

We use Average Precision@$K$, defined by Definition~\ref{avg-preK}, to measure the effectiveness of studied role similarity metrics in top-k similarity search.
\begin{definition}\label{avg-preK}
Let $G=(V,E)$ be a labeled network with $n$ nodes, and let \textsf{Sim} be a node similarity measure. For a node $u \in V$, let $S_{u,k}$ be the set of the top-k similar nodes for $u$ based on \textsf{Sim}, and let $k_u^+$ be the number of nodes that have the same label as $u$ in $S_{u,k}$. Thus, for \textsf{Sim}, its Average Precision@$K$ on $G$ is defined as 
\begin{align*}
	{\rm  Average \; Precision@}K 
	:= \frac{1}{n} \displaystyle\sum_{u \in V} \frac{k_u^+}{k}.
\end{align*}
\end{definition}
Figure~\ref{AvgPreK} shows the Average Precision@$K$ of four studied measures on six labeled networks. Note that the difference between \textsf{ForestSim-EX} and \textsf{ForestSim-AP} is marginal though the latter requires less time and space. Moreover, \textsf{ForestSim} achieves comparable performance to \textsf{RoleSim}, the state-of-art role similarity metric, and it clearly outperforms \textsf{StructSim} in most cases.

In brief, \textsf{ForestSim} achieves much better performance than \textsf{StructSim}, and the difference between \textsf{ForestSim} and \textsf{RoleSim} is marginal in terms of effectiveness.

\section{Conclusion}~\label{S-Conclude}
In this paper, we propose a novel node similarity metric, namely \textsf{ForestSim}, to quickly and effectively process top-k similarity search on large networks. Different from previous frameworks, \textsf{ForestSim}, based on spanning rooted forests of graphs, adopts the average size of all trees rooted at node $u$ in spanning rooted forests to reflect the structural information of $u$. We also provide the top-k similarity search algorithm, namely \textsf{ForestSimSearch}, that can handle a top-k query in $O(k)$ time once the precomputation is finished. Furthermore, we use the fast approximate algorithm to compute the diagonal entries of the forest matrix so that we can finish the precomputation within nearly linear time and nearly linear space. Extensive experimental results on the real-world networks demonstrate that our proposed method works efficiently on large networks and achieves comparable performance to the state-of-art measures.

\section*{DATA AVAILABILITY STATEMENT}

The data underlying this article are available in Network Repository, at https://networkrepository.com.

\bibliographystyle{compj}
\bibliography{Forest}

\end{document}